\pdfoutput=1
\RequirePackage{latexml}
\iflatexml
	\documentclass{article}
	\author{
		Dominik Peters
		\\
		CNRS, LAMSADE, Universit\'e Paris Dauphine - PSL
		\\
		dominik.peters@lamsade.dauphine.fr
	}
	
	\usepackage{algorithm}
	\usepackage{algpseudocode}
\else
	\documentclass[11pt]{scrartcl}
	\usepackage[a4paper, total={16cm, 24cm}]{geometry}
	\renewcommand\tableofcontents{\listoftoc*{toc}} % no "Contents" header
	\KOMAoptions{toc=bibnumbered}
	\setcapindent{0pt}
	\setkomafont{captionlabel}{\sffamily\bfseries}

	\usepackage{authblk}

	\author[1]{Dominik Peters}
	\affil[1]{CNRS, LAMSADE, Universit\'e Paris Dauphine - PSL}

	\usepackage{multicol}
	\usepackage{tocloft}

	\usepackage[small]{caption}
	
	\usepackage{algorithm}
	\usepackage{algpseudocodex}
\fi

\usepackage{natbib}
\usepackage[pagebackref]{hyperref}
\hypersetup{
	pdfencoding=auto, 
	psdextra,
	colorlinks=true,
	citecolor=green!50!black,
	linkcolor=red!60!black,
	urlcolor=purple!70!black
}

\usepackage{url}
\usepackage{microtype}
\usepackage[utf8]{inputenc}
\usepackage{graphicx}
\usepackage{amsmath}
\usepackage{amsthm}
\usepackage{amssymb}
\usepackage{mathtools}
\usepackage{booktabs}
\usepackage{enumitem}

\usepackage{wrapfig}
\iflatexml\else
\makeatletter
\patchcmd\WF@putfigmaybe{\lower\intextsep}{}{}{\fail}%
\AddToHook{env/wrapfigure/begin}{\setlength{\intextsep}{0pt}}
\makeatother
\fi

\usepackage[nameinlink]{cleveref}
\iflatexml\else
\usepackage{crossreftools}
\fi
\usepackage{stmaryrd}
\usepackage{thm-restate}

\usepackage{tikz}
\renewcommand{\epsilon}{\varepsilon}
\renewcommand*{\le}{\leqslant}

\renewcommand*{\ge}{\geqslant}

\newtheoremstyle{sfthm}% sans serif style
	{\topsep}% Space above
	{\topsep}% Space below
	{\itshape}% Body font
	{}% Indent amount
	{\sffamily\bfseries}% Theorem head font
	{}%Punctuation after theorem head
	{.5em}%Space after theorem head
	{}% theorem head spec
\theoremstyle{sfthm}
\newtheorem{theorem}{Theorem}[section]
\newtheorem{lemma}[theorem]{Lemma}
\newtheorem{proposition}[theorem]{Proposition}

\newtheoremstyle{sfdef}% sans serif style
{\topsep}% Space above
{\topsep}% Space below
{}% Body font
{}% Indent amount
{\sffamily\bfseries}% Theorem head font
{}%Punctuation after theorem head
{.5em}%Space after theorem head
{}% theorem head spec
\theoremstyle{sfdef}
\newtheorem{definition}[theorem]{Definition}
\newtheorem{example}[theorem]{Example}
\newtheorem{remark}[theorem]{Remark}

\newcommand{\A}{\mathcal{A}}
\newcommand{\PAVscore}{\textup{\textsf{PAV-score}}}

\title{The Core of Approval-Based Committee Elections with Few Seats}

\date{\vspace{-5pt}\sffamily\normalsize Full version of the paper appearing in the proceedings of IJCAI 2025 $\cdot$ May 2025}

\usepackage{multicol}

\begin{document}

\maketitle

\begin{abstract}
	\iflatexml\else
	\begin{center}
		\textbf{\textsf{Abstract}} \smallskip
	\end{center}\fi
	In an approval-based committee election, the goal is to select a committee consisting of $k$ out of $m$ candidates, based on $n$ voters who each approve an arbitrary number of the candidates. The \emph{core} of such an election consists of all committees that satisfy a certain stability property which implies proportional representation. In particular, committees in the core cannot be ``objected to'' by a coalition of voters who is underrepresented. The notion of the core was proposed in 2016, but it has remained an open problem whether it is always non-empty. We prove that core committees always exist when $k \le 8$, for any number of candidates $m$ and any number of voters $n$, by showing that the Proportional Approval Voting (PAV) rule due to \citet{Thie95a} always satisfies the core when $k \le 7$ and always selects at least one committee in the core when $k = 8$. We also develop an artificial rule based on recursive application of PAV, and use it to show that the core is non-empty whenever there are $m \le 15$ candidates, for any committee size $k \le m$ and any number of voters $n$. These results are obtained with the help of computer search using linear programs.
\end{abstract}

\iflatexml\else
\vspace{15pt}

\hrule

\vspace{5pt}
{
	\setlength\columnsep{35pt}
	\setcounter{tocdepth}{2}
	\renewcommand\contentsname{\vspace{-20pt}}
		{\small
			\tableofcontents}
}

\vspace{11pt}
\hrule

\newpage
\fi

\section{Introduction}\label{sec:introduction}

The seminal work of \citet{ejr-v4,ejr} introduced a rigorous way of reasoning about voter representation in multi-winner elections. Their model considers \emph{approval-based committee elections}, where the task is to identify a committee $W \subseteq C$ of $k$ out of $m$ candidates, based on a set of voters~$N$, with each $i \in N$ indicating a subset $A_i \subseteq C$ of candidates that $i$ approves. \citet{ejr-v4,ejr} formulated a compelling axiom called Extended Justified Representation (EJR) which gives representation guarantees to every group of voters who approve sufficiently many candidates in common. Researchers discovered that voting rules developed 130 years ago by \citet{Thie95a} and \citet{Phra94a} satisfy this or related axioms \citep{brill2023phragmen,ejr,Jans16a}. Interesting \emph{new} rules satisfying EJR have recently been developed \citep{aziz2018complexity,PeSk20,ejr+}, with one of them (the ``Method of Equal Shares'') now in active use for the participatory budgets in several cities in Poland, Switzerland, and the Netherlands.

\citet[Section 5.2]{ejr-v4,ejr} also defined another representation axiom that is significantly stronger than EJR, called \emph{core stability} in analogy to a similar concept from cooperative game theory. A committee $W$ is core stable if for every set $T \subseteq C$, there are not too many voters who prefer the set $T$ to $W$, namely we have
\[
\big|\big\{ i \in N : |A_i \cap T| > |A_i \cap W| \big\}\big| < |T| \cdot \frac{n}{k}.
\]
If this inequality were violated for some $T$, then the set of voters on the left-hand side could form a \emph{blocking coalition} of a size that is large enough for the coalition to ``deserve'' to decide to include $T$ in the committee.

The EJR property is weaker than core stability (because under EJR voters are only allowed to join the blocking coalition if they approve all the candidates in $T$, i.e., $|A_i \cap T| = |T|$), but in exchange there are several attractive voting rules satisfying EJR. On the other hand, \citet{ejr-v4,ejr} noted that ``the core stability condition appears to be too demanding, as none of the voting rules considered in our work is guaranteed to produce a core stable outcome''. No such voting rules have been discovered since. They conclude: ``It remains an open question whether the core [is always] non-empty.'' This question remains open more than 8 years later.

Some amount of progress has been made, and in particular it is known that there always exist committees satisfying approximate variants of the core \citep{PeSk20,jiang2020approximatelystable,munagala2022coremultilinear}, and the core exists on single-peaked approval profiles \citep{pierczynski2022restricteddomains} and on profiles where each candidate has at least $k$ copies \citep{approvalbasedapportionment}.

To the best of my knowledge, the only known existence result that holds \emph{in general} is that the core is non-empty for $k = 3$, which \citet[Section 3.1]{cheng2020group} showed by case analysis. \citet{cheng2020group} conclude that a ``major open question is the existence of deterministic stable committees in the Approval Set setting, generalizing our positive result for $k = 3$ to general $k$. We conjecture that such a stable committee always exists. Via computer-assisted search, we have shown that this conjecture holds for small numbers of voters and candidates ($m + n \le 14$).''

It might seem surprising that the state of the art has not improved beyond these very small parameter values. In particular, there is a natural way of using mixed integer linear programming (ILPs) to search for counterexamples to core existence: fix $m$ and $k$, and introduce a fractional variable for each possible ballot, indicating what fraction of the voters submit this ballot. Then, for every possible committee, enforce using binary variables that there exists at least one successful core deviation. If an ILP solver determines that the resulting program is infeasible, this implies the non-emptiness of the core for $m$ and $k$, for any number $n$ of voters. Unfortunately, the size of this program grows rapidly, and they are not easy to solve even for very small sizes (Gurobi solves $m = 7$, $k = 5$, in 450s, but did not solve $m=7$, $k=4$ after 134\,000s (37h) on 8 cores).% 
\footnote{The same problem can also be encoded as an SMT problem on linear arithmetic. This can sometimes lead to faster solve times, but it is also only feasible for very small sizes.}

In this paper, by deriving a new way of using solvers, we show that the core always exists for committee sizes up to $k = 8$, regardless of the number of candidates (improving upon the previous result for $k = 3$). We also show that the core always exists when the number $m$ of candidates is at most $15$, for any $k \le m$. Both results hold \emph{for any number $n$ of voters}.

These results are obtained by analyzing variants of \emph{Proportional Approval Voting} (PAV), the voting rule proposed by \citet{Thie95a}. This voting rule works by maximizing a carefully chosen objective function over the set of all committees of size $k$. We show via linear programs that PAV always selects a core-stable committee when $k \le 7$, and that it always selects at least one core-stable committee when $k = 8$, perhaps tied with other committees that fail core-stability. This performance of the PAV rule is remarkable, and contrasts with other rules such as the Phragm\'en rule and the Method of Equal Shares which fail the core for $k=6$ and $k=7$, respectively (see \Cref{fn:phragmen-mes}). In many applications, the number of seats to fill will be 8 or fewer, and thus our result suggests that PAV is a good rule for such contexts.%
\footnote{Note, however, that even for $k = 6$, PAV may select committees that fail laminar proportionality \citep{PeSk20} and are intuitively unfair. A three-voter example for $k = 6$ is $A_1 = \{c_1, c_2, c_3\}$, $A_2 = \{c_1, c_2, c_4\}$, and $A_3 = \{c_5, c_6, c_7\}$ where $W = \{c_1, c_2, c_3, c_5, c_6, c_7\}$ is a global PAV committee (tied with others) while laminar proportionality demands that $\{c_1, c_2, c_3, c_4\} \subseteq W$.}

For the results about limited numbers of candidates, we consider a recursive version of PAV where, if the PAV committee fails the core because some set $T$ has too much support, we then re-compute PAV subject to the constraint that $T \subseteq W$ and without taking into account the voters who were part of the blocking coalition. If the result still fails core-stability, we add additional constraints. We show that if $m \le 15$, this process always terminates with a core-stable committee. However, there are examples where the rule fails to guarantee the core for $m = 16$.

As mentioned, these results were obtained with the help of linear programming. This becomes feasible even for these quite large sizes because we can fix \emph{one} committee, and add constraints that this committee is the one selected by the voting rule under consideration. This is much simpler than a program that needs constraints for all possible committees. Linear programming has been used before to analyze sequential versions of PAV \citep{skowron2021proportionality,pjr}.
The infeasibility of the relevant programs can be compactly certified via Farkas' lemma, allowing efficient verification of our results without having to trust a solver. Code for these tasks is available at \url{https://github.com/DominikPeters/core-few-candidates/}.

\section{Related Work}

\paragraph{Barriers to core existence.}
Proving that the core is non-empty is difficult because several natural strategies are known not to work.
Importantly, all known voting rules that satisfy weakenings of the core such as EJR fail the core, including the PAV rule (\citealp[Example 6]{ejr}, \citealp[Section 1]{PeSk20}). \citet[Theorem 10]{PeSk20} show that every welfarist rule (one that depends only on voter utilities) must fail the core. They also show that every voting rule satisfying the Pigou--Dalton principle (which says that outcomes that induce a more equitable social welfare distribution should be preferred; this is satisfied by PAV) cannot satisfy the core, and indeed cannot provide better than a 2-approximation to it \citep[Theorem 5]{PeSk20}.

\paragraph{Computational complexity.}
It is NP-hard to compute the PAV rule \citep[Corollary 1]{aziz2015computational}, but it is fixed-parameter tractable for a variety of parameters \citep{yang2023parameterized}. A local search variant of PAV retains its proportionality properties and can be computed in polynomial time \citep{aziz2018complexity} for an appropriately chosen tolerance parameter \citep{kraiczy2024localsearch}. The problem of checking whether a given committee is in the core is coNP-complete \citep[Theorem 5.3]{approvalbasedapportionment} and remains hard even when every voter approves at most 6 candidates, and each candidate is approved by at most 2 voters \citep[Theorem 1]{munagala2022auditing}. The verification problem is also hard to approximate to within a factor better than $1 + 1/e$ \citep[Theorem 2]{munagala2022auditing}, though a logarithmic approximation algorithm exists \citep[Theorem 3]{munagala2022auditing}.

\paragraph{More general settings.}
We work in the model of approval-based committee elections. If non-approval preferences are allowed (such as cardinal additive valuations), the core may be empty (\citealp[Appendix C]{fain2018fair}; \citealp[Example 2]{peters2021proportional}). For participatory budgeting applications, one can replace the cardinality constraint in the definition of a committee by a knapsack constraint. For this non-unit cost setting, approval votes have several interpretations \citep[Section 3.4.2]{rey2023computationalsocialchoiceindivisible}. One option is \emph{cost utilities}, which measures a voter's utility by the total cost of approved winning projects. For this utility model, the core may be empty \citep{maly2023coreapprovalbasedpbinstance}. For \emph{cardinality utilities}, where the voter's utility is the number of approved winning projects, the non-emptiness of the core is an open question.

\paragraph{Approximate core.}
Core stability can be relaxed through multiplicative approximations, in two main ways. Say that a committee $W$ is in the $(\alpha,\beta)$-core, $\alpha, \beta \ge 1$, if for every potential deviation $T$, we have $|\{ i \in N : |A_i \cap T| > \alpha \cdot |A_i \cap W| \}| < \beta \cdot |T| \cdot \frac{n}{k}$. For $ \alpha = \beta = 1$, this is the core; for $\alpha > 1$, every member of the blocking coalition must increase their utility by a factor of $\alpha$; for $\beta > 1$, we require that blocking coalitions must be larger than usual by a factor of $\beta$.

\citet{PeSk20} show that PAV is in the $(2,1)$-core, and that the Method of Equal Shares is in the $(\log k, 1)$-core (or a mild relaxation of that concept).
\citet{munagala2022coremultilinear} show that the $(9.27,1)$-core is non-empty even for general additive valuations. \citet[Appendix C]{fain2018fair} show that the $(1 + \epsilon, 1)$-core is non-empty if it is additionally additively relaxed.
\citet{jiang2020approximatelystable} show that the $(1,16)$-core is non-empty, via rounding a stable lottery \citep{cheng2020group}; they conjecture that at least the $(1,2)$-core is always non-empty. A similar rounding technique has been used to analyze Condorcet winning sets \citep{charikar2024candidatessufficewinvoter}.

\paragraph{Fractional models.}
Analogs of core-stability have also been defined in fractional models. For example, \citet{ABM20a} consider ``fair mixing'' in an approval-based model, where the output is a probability distribution over candidates (which can be interpreted as a division of a budget). They show that the rule maximizing Nash welfare (which is related to the PAV rule) satisfies core-stability. \citet{FGM16a} obtain the same result in a more general model with additive linear utilities. They also show that a core-stable outcome exists for fractional committees (which can be viewed as a probability distribution where each candidate receives mass at most $1/k$) via Lindahl equilibrium which is known to exist from fixed-point theorems \citep{foley1970lindahl,munagala2022coremultilinear} and via convex programming \citep{kroer2025lindahl}. Various weakenings of the core have also been studied in approval-based fair mixing and related models \citep[see, e.g.,][]{brandl2021distribution,suzuki2024maxflow,bei2024truthful}.

\section{Preliminaries}\label{sec:preliminaries}

Let $C = \{c_1, \dots, c_m\}$ be a set of $m$ candidates. Let $k \in \{1, \dots, m\}$ be the committee size. A \emph{committee} is a subset $W \subseteq C$ of winning candidates with $|W| = k$.

Let $\mathcal A$ denote the set of non-empty subsets of $C$, which we will refer to as \emph{approval sets} or \emph{ballots}. A \emph{profile} is a map $P : \mathcal A \to \mathbb Q_{\ge 0}$ with $\sum_{A \in \A} P(A) = 1$, where $P(A)$ indicates the fraction of the voters that approve exactly the candidates in~$A$. Given an approval set $A \in \A$ and a committee $W$, the \emph{utility} of the committee for $A$ is the number of approved committee members: $u_A(W) = |A \cap W|$.

A nonempty set $T \subseteq C$ with $|T| \le k$ is called a \emph{potential deviation}. Given a profile $P$, a committee $W$ is \emph{core stable} if for every potential deviation $T$, we have \[
\sum_{A \in \A : u_A(T) > u_A(W)} P(A) < \frac{|T|}{k}.
\]
Otherwise, $T$ is called a \emph{deviation} (or a \emph{successful deviation}) from $W$. The \emph{core} is the set of committees that are core stable.

The $n$th \emph{harmonic number} is $H(n) = 1 + \frac12 + \frac13 + \cdots + \frac1n$. The \emph{PAV score} assigned to a committee $W$ by ballot $A \in \A$ is
\[
\PAVscore_A(W) = H(u_A(W)).
\]
Given a profile $P$, the \emph{PAV score} of $W$ under profile $P$ is
{\setlength{\belowdisplayskip}{3pt}
\[
\PAVscore_P(W) = \sum_{A \in \A} P(A) \cdot \PAVscore_A(W).
\]
A committee is called a \emph{global PAV committee} if it has the highest PAV score.
The \emph{PAV rule} selects all global PAV committees.}

We will also be interested in local maxima of the PAV objective, i.e., committees that have a weakly higher PAV score than any committee obtained by performing a single swap. 
For a committee~$W$, $x \in W$ and $y \in C \setminus W$, we write $W_{xy} = W \setminus \{x\} \cup \{y\}$ for the committee obtained by replacing $x$ with $y$. Given a profile $P$ and a fixed committee $W$, we write 
\[
\Delta_{P,x,y} = \PAVscore_P(W_{xy}) - \PAVscore_P(W)
\]
for the increase in PAV score resulting from this swap. For a ballot $A$, we define $\Delta_{A,x,y}$ analogously. Then we say that $W$ is a \emph{local PAV committee} if $\Delta_{P,x,y} \le 0$ for all $x$ and $y$. Note that every global PAV committee is also a local PAV committee.

The following lemma shows that a local PAV committee never admits deviations of certain kinds: neither disjoint deviations, nor deviations that contain only a single unelected candidate.

\begin{lemma}
	\label{lem:pav-special-core-deviations}
	Let $W$ be a local PAV committee, and let $T$ be a potential deviation. Suppose we have
	\[
	\text{(i) } T \cap W = \emptyset, 
	\qquad \text{ or } \qquad
	\text{(ii) } |T \setminus W| \le 1.
	\]
	 Then $T$ is \emph{not} a deviation from $W$.
\end{lemma}
\begin{proof}
	(i) is proved by \citet[Theorem 3.2 and Remark 3.4]{approvalbasedapportionment} using a swapping argument. 
	
	(ii) If $T \setminus W = \emptyset$, then $T \subseteq W$, and there are no approval ballots that strictly prefer $T$ to $W$. If $T \setminus W = \{c\}$, then for every $A \in \A$ with $u_A(T) > u_A(W)$, we have $c \in A$. Writing $\ell = |T|$, if $T$ were a deviation, we would thus have a set $S$ of ballots forming $\ell/k$ of the profile, who all approve $c \not\in W$, and who all have utility $u_A(W) < u_A(T) \le \ell$. Thus, we have a violation of the EJR+ axiom of \citet{ejr+} which local PAV committees are known to satisfy.
\end{proof}

Note that \Cref{lem:pav-special-core-deviations}(ii) implies that PAV satisfies the core when $k = m - 1$ or $k = m$.

We recall the following well-known result about systems of linear inequalities, providing a certificate of infeasibility.

\begin{lemma}[Farkas' Lemma]
	Let $A \in \mathbb Q^{mn}$ be an $m \times n$ matrix, and let $b \in \mathbb Q^m$. Then the following are equivalent.
	\begin{enumerate}
		\item[(i)] There does not exist $x \in \mathbb Q^n$ with $Ax \le b$. 
		\item[(ii)] There exists an integer vector $y \in \mathbb Z_{\ge 0}^m$ such that $y^T b < 0$ and $A^T y \ge 0$.
	\end{enumerate}
\end{lemma}

\noindent
Thus, by exhibiting the integer vector $y$, one can prove the infeasibility of the system $Ax \le b$.

\section{Small Committee Size}\label{sec:small-k}

In this section, we discuss core-stable committees when the committee size $k$ is small. We give existence results when $k \le 8$, separately handling the cases $k \le 7$ and $k = 8$.

\subsection{Committee size $k \le 7$}

For committee sizes up to $7$, core-stable committees always exist because every local PAV committee is in the core. This establishes existence of core-stable committees, but also indicates that PAV rule is a very good rule for smaller committee sizes.%
\footnote{\label{fn:phragmen-mes}There are examples where the sequential Phragm\'en rule fails core (and even EJR) for $k = 6$, and where MES fails core for $k = 7$. These counterexamples work even for the party-approval setting \citep{approvalbasedapportionment}, where each candidate can be placed in the committee several times. For Phragm\'en, take the 3-voter profile $(ab, bc, ac)$, where Phragm\'en can elect $ababab$, with $T = \{c,c,c,c\}$ forming a deviation. For MES, take the 7-voter profile $(ab,ac,ad,bcd,bcd,bcd,bcd)$, where $bbbbbaa$ is an outcome of MES, with $T = \{c,c,c,d,d,d\}$ forming a deviation. (In these examples, there are other tied outcomes in the core, and I don't know if unique examples exist.)}
Like other proofs that PAV is proportional, our proof reasons about the change in PAV score caused by certain swaps. To establish a key inequality, the proof refers to the result of a computer enumeration of all possible ballot types, which can be reproduced using \href{https://github.com/DominikPeters/core-few-candidates/blob/master/A_committee_size_7/check-inequality.ipynb}{Python code available on GitHub}.

\newcommand{\coalitionsize}{s}
\begin{theorem}
	\label{thm:pav-k7}
	\!\!When $k \le 7$, every local PAV committee is in the core.
\end{theorem}
\begin{proof}
	Let $W$ be a PAV committee with $|W| = k \le 7$, and consider a potential deviation $T$. We need to show that $T$ is not successful. Writing
	\iflatexml
		\[
		\coalitionsize = \sum_{A \in \A : u_A(T) > u_A(W)} P(A),
		\]
	\else
		\[
			\coalitionsize = \quad\sum_{\mathclap{\substack{A \in \A \\ u_A(T) > u_A(W)}}}\quad P(A),
		\]
	\fi
	we need to show that $\coalitionsize < |T|/k$. We may assume that $\coalitionsize > 0$, since otherwise $T$ is definitely not successful.
	
	We now consider the following process: we remove some candidate $x \in W \setminus T$ from the committee $W$, and replace it by some candidate $y \in T \setminus W$, and will estimate the total effect of these swaps. Because $W$ is a local PAV committee, we have
	\begin{equation}
		\label{eq:summed differences in PAV score}
		\sum_{x \in W \setminus T} \sum_{y \in T \setminus W} \Delta_{P,x,y} \le 0
	\end{equation}
	We can rewrite \eqref{eq:summed differences in PAV score} by computing the contributions of each ballot to the differences in PAV score. Let $A \in \A$. Note that swapping in $y$ for $x$ either increases the utility $u_A(W)$ by 1 point (if $x \not\in A$ but $y \in A$), decreases it by 1 point (if $x \in A$ but $y \not \in A$), or otherwise it stays the same. The number of $(x,y)$ pairs leading to an increase is $|(W \setminus T) \setminus A| \cdot |(T \setminus W) \cap A|$; the number of pairs leading to a decrease is $|(W \setminus T) \cap A| \cdot |(T \setminus W) \setminus A|$. 
	Thus, from the point of view of the ballot $A$, the total difference in PAV score summed across all $x$ and $y$ becomes
	\begin{align*}
		\delta_A :\!&= \sum_{x \in W \setminus T} \sum_{y \in T \setminus W} \Delta_{A,x,y} \\
		&= \frac{|(W \setminus T) \setminus A| \cdot |(T \setminus W) \cap A|}{u_A(W)+ 1}
		- \frac{|(W \setminus T) \cap A| \cdot |(T \setminus W) \setminus A|}{u_A(W)}.
	\end{align*}
	Thus we can rewrite \eqref{eq:summed differences in PAV score} as follows:
	\begin{equation}
		\label{eq:combinatorial differences sum}				
		\sum_{A \in \A} P(A) \cdot \delta_A \le 0
	\end{equation}
	We will now give lower bounds on the value of $\delta_A$.
	
	First we give a trivial lower bound, valid for all $A \in \A$:
	\begin{align*}
		\delta_A 
		&\ge - \frac{|(W \setminus T) \cap A| \cdot |(T \setminus W) \setminus A|}{u_A(W)}
		\tag{dropping positive terms}\\
		&\ge - \frac{|W \cap A| \cdot |T \setminus W|}{u_A(W)} \\
		&= - |T \setminus W|. \tag{since $u_A(W) = |W\cap A|$}
	\end{align*}
	Next, suppose $A \in \A$ is a ballot with $u_A(T) > u_A(W)$, i.e., $|A \cap T| \ge |A \cap W| + 1$. 
	For such ballots, there exists a better lower bound for $\delta_A$. In particular, it is the case that
	\begin{equation}
		\label{eq:ballot condition}
		\delta_A > (\tfrac{k}{|T|} - 1) \cdot |T \setminus W|.
	\end{equation}
	
	This can be shown by a \href{https://github.com/DominikPeters/core-few-candidates/blob/master/A_committee_size_7/check-inequality.ipynb}{small exhaustive search} over all possible combinations of the quantities appearing in the definition of $\delta_A$. In particular, write
	\[
	a = |(W \setminus T) \cap A|, \quad
	b = |(W \cap T) \cap A|, \quad
	c = |(T \setminus W) \cap A|.
	\]
	Then we have $\delta_A = \frac{(|W \setminus T| - a) c}{a + b + 1} - \frac{a (|T\setminus W| - c)}{a + b}$, and one can check that this is always strictly larger than $(\tfrac{k}{|T|} - 1) \cdot |T \setminus W|$ by trying all triples $(a, b, c)$ with $0 \le a \le |W \setminus T|$, $0 \le b \le |W \cap T|$, and $0 \le c \le |T \setminus W|$ which satisfy $b + c > a + b$ (which encodes that $u_A(T) > u_A(W)$).%
	\footnote{Note that none of $a$, $b$, $c$ depend on $m$, so the search is finite and the proof works independently of the number of candidates.}
	
	Based on these bounds, we have
	\begin{align*}
		0 
		&\ge \textstyle\sum_{A \in \A} P(A) \cdot \delta_A 
		\tag{by \eqref{eq:combinatorial differences sum}} \\
		&> \coalitionsize \big( (\tfrac{k}{|T|} - 1) \cdot |T \setminus W|\big) + (1 - \coalitionsize) (- |T \setminus W|) \tag{$\coalitionsize > 0$} \\
		&= |T \setminus W| \cdot (\coalitionsize (\tfrac{k}{|T|} - 1) - (1 - \coalitionsize)) \\
		&= |T \setminus W| \cdot (\coalitionsize \tfrac{k}{|T|} - 1).
	\end{align*}
	Thus, because $|T \setminus W| > 0$, we have $0 > \coalitionsize \smash{\frac{k}{|T|}} - 1$ and hence $\coalitionsize < \frac{|T|}{k}$, as desired.
\end{proof}

While it may not look like it, the result of \Cref{thm:pav-k7} was obtained by linear programming. Suppose that the result is false, so that there exists a profile $P$ and a local PAV committee $W$ such that some $T$ is a successful deviation. Without loss of generality, there exists such an example with $W = \{c_1, \dots, c_k\}$. Note that $P$ then forms a solution to the following system of linear inequalities, where we may assume that $C = W \cup T$:%
\footnote{We may make this assumption because a counterexample on a larger $C$ remains a counterexample when restricted to $W \cup T$, because a local PAV committee remains a local PAV committee after deleting candidates outside the committee.}
\begin{equation}
	\label{program:PAV counterexample}
	\begin{aligned}
		\sum_{A \in \A} P(A) &= 1 \\
		\Delta_{P,x,y} &\le 0 && \text{for all $x \in W$ and $y \in C \setminus W$} \\
		\sum_{\mathclap{A \in \A : u_A(T) > u_A(W)}} \: P(A) &\ge \frac{|T|}{k} \\
		P(A) &\ge 0 && \text{for all } A \in \A
	\end{aligned}
\end{equation}
Thus, \Cref{thm:pav-k7} is proven if the system \eqref{program:PAV counterexample} does \emph{not} have a feasible solution, for all potential deviations $T$. 
By Farkas' lemma, that means that for every $T$, there exist
$\alpha \in \mathbb{R}$, $(\beta_{xy})_{xy} \ge 0$, $\gamma \ge 0$ satisfying
\begin{align*}
	\alpha - \tfrac{|T|}{k}\gamma < 0 \\
	\alpha + \sum_{x\in W}\sum_{y\not\in W} \Delta_{A,x,y} \cdot \beta_{xy} - \gamma \ge 0 && \text{for all } A \in \A \text{ with } u_A(T) > u_A(W) \\
	\alpha + \sum_{x\in W}\sum_{y\not\in W} \Delta_{A,x,y} \cdot \beta_{xy} \ge 0 && \text{for all } A \in \A \text{ with } u_A(T) \le  u_A(W)
\end{align*}
In fact, the proof of \Cref{thm:pav-k7} simply constructs such a certificate solution for every $T$, where $\alpha = |T \setminus W|$ and $\beta_{xy} = 1$ whenever $x \in W \setminus T$ and $y \in T \setminus W$, and $\beta_{xy} = 0$ otherwise. \\

Due to the following remark, it is possible to compute a core-stable outcome in $O(m^2n)$ time whenever $k \le 7$.

\begin{remark}
	Let $k \le 7$ and $\epsilon = 0.1/k^2$. By solving linear programs \href{https://github.com/DominikPeters/core-few-candidates/blob/master/A_committee_size_7/pav-almost-swap-stable.py}{[GitHub]}, one can check that every committee that is $\epsilon$-local-swap-stable (i.e., $\Delta_{P,x,y} \le \epsilon$ for all $x$ and $y$) is in the core. An $\epsilon$-local-swap-stable and thus a core-stable committee can be found by performing $O(k^2 \ln k) = O(1)$ many $\epsilon$-improving swaps \citep[see][Proposition~1]{aziz2018complexity}, each of which takes $O(m^2n)$ time to find. It turns out that for $\epsilon = 1/k^2$, it is not the case that every $\epsilon$-local-swap-stable committee is in the core, even though this value of $\epsilon$ is enough to ensure that the committee satisfies EJR \citep[Theorem 1]{aziz2018complexity}. For example, in the profile with $P(\{a,b\}) = P(\{a,c\}) = 0.25$ and $P(\{d,e,f,g,h\}) = 0.5$, for $k = 6$, the committee $\{a,d,e,f,g,h\}$ fails the core due to $T = \{a,b,c\}$, but it is $1/40$-local-swap-stable, and $1/40 < 1/36 = 1/k^2$.
\end{remark}

\subsection{Committee size $k = 8$}

\Cref{thm:pav-k7} does not hold for $k = 8$: There are profiles where some local (and even global) PAV committee is not in the core.

\begin{example}[PAV may fail core for $k = 8$]
	\label{ex:pav-k8}
	Consider an instance with 4 voters, $v_1$ approving $\{c_1, c_2, c_3\}$, and $v_2$ approving $\{c_1, c_2, c_4\}$, and the other 2 voters approving $\{c_5, c_6, c_7, c_8, c_9, c_{10}\}$. This profile is depicted below, where each voter approves the candidates above the voter's label.
	\[
	\begin{tikzpicture}
		[yscale=0.4,xscale=0.78,voter/.style={anchor=south, yshift=-7pt}, select/.style={fill=blue!10}, c/.style={anchor=south, yshift=1.5pt, inner sep=0}]
		\draw[select] (0,0) rectangle (2,1);
		\draw[select] (0,1) rectangle (2,2);
		\draw (0,2) rectangle (1,3);
		\draw (1,2) rectangle (2,3);
		
		\draw[select] (2,0) rectangle (4,1);
		\draw[select] (2,1) rectangle (4,2);
		\draw[select] (2,2) rectangle (4,3);
		\draw[select] (2,3) rectangle (4,4);
		\draw[select] (2,4) rectangle (4,5);
		\draw[select] (2,5) rectangle (4,6);
		
		% Candidate labels
		\node at (1,0.42) {$c_1$};
		\node at (1,1.42) {$c_2$};
		\node at (0.5,2.42) {$c_3$};
		\node at (1.5,2.42) {$c_4$};
		\node at (3,0.42) {$c_5$};
		\node at (3,1.42) {$c_6$};
		\node at (3,2.42) {$c_7$};
		\node at (3,3.42) {$c_8$};
		\node at (3,4.42) {$c_9$};
		\node at (3,5.42) {$c_{10}$};
		
		% Voter labels
		\foreach \i in {1,...,4}
		\node[voter] at (\i-0.5,-1) {$v_\i$};
	\end{tikzpicture}
	\]
	On this profile, $W = \{c_1, c_2, c_5, c_6, c_7, c_8, c_9, c_{10}\}$ is a global PAV committee (indicated in blue in the picture). However, $W$ is not in the core: consider $T = \{c_1, c_2, c_3, c_4\}$, which has support from $\frac12$ of the voters, and $|T|/k = \frac12$.
	\qed
\end{example}

Note, however, that in \Cref{ex:pav-k8}, there is more than one global PAV committee. In particular, $W' = \{c_1, c_2, c_3, c_5, c_6, c_7, c_8, c_9\}$ (obtained by removing $c_{10}$ and adding $c_3$) is also a global PAV committee and it is in the core.

It turns out that \Cref{ex:pav-k8} is essentially the only example where a PAV committee fails to be core-stable for $k = 8$, as all such examples share the same structure.

\begin{lemma}
	\label{lem:pav-k8-counterexamples}
	Let $P$ be a profile and suppose that $W$ with $|W| = 8$ is a local PAV committee that is not in the core due to objection $T$. Then there exist distinct $a,b \in W$ and distinct $x, y \in C \setminus W$ such that $T = \{a,b,x,y\}$. In addition,
	\begin{itemize}
		\item[(i)] one quarter of the voters submit ballots $A$ such that $A \cap (W \cup T) = \{a, b, x\}$ and another quarter submit ballots with $A \cap (W \cup T) = \{a, b, y\}$,
		\item[(ii)] the remaining half of the voters submit ballots that are disjoint from $T$, and
		\item[(iii)] the PAV score of $W$ is reduced by exactly $1/12$ if any one member of $W \setminus \{a,b\}$ is removed.
	\end{itemize}
\end{lemma}
\begin{proof}
	We first check that if $W$ is  not in the core, then any core objection must use a $T$ with $|T|=4$ and $|W \cap T| = 2$. This can be deduced using the linear programming approach behind \Cref{thm:pav-k7}; by iterating through all possible $T$, we find that the system \eqref{program:PAV counterexample} has a solution only for $T$ satisfying the condition in the theorem statement. Alternatively, one can check \href{https://github.com/DominikPeters/core-few-candidates/blob/master/A_committee_size_7/check-inequality.ipynb}{[GitHub]} that the inequality \eqref{eq:ballot condition} is only violated for such $T$; thus for other $T$ the proof of \Cref{thm:pav-k7} goes through.
	
	Now fix such a $T = \{a,b,x,y\}$. Assume that there exists a profile $P$ where $W$ is a local PAV committee with successful deviation $T$ but that violates any of the conditions (i)--(iii). Then if we delete all candidates outside $W \cup T$ from the profile, it would still fail (i)--(iii). Thus, for purposes of making the following linear programs finite, we may assume that $C = W \cup T$ (so $|C| = 10$).
	
	To prove (i), we solve the following four linear programs:
	\begin{align*}
		\text{maximize }& P(\{a,b,x\}) \: \text{ subject to $P$ satisfying \eqref{program:PAV counterexample}} \\
		\text{minimize }& P(\{a,b,x\}) \: \text{ subject to $P$ satisfying \eqref{program:PAV counterexample}} \\
		\text{maximize }& P(\{a,b,y\}) \: \text{ subject to $P$ satisfying \eqref{program:PAV counterexample}} \\
		\text{minimize }& P(\{a,b,y\}) \: \text{ subject to $P$ satisfying \eqref{program:PAV counterexample}}
	\end{align*}
	The optimal solutions to all these programs is $\frac14$.
	
	To prove (ii), iterate through all ballots $A \in \A$ with $A \cap T \neq \emptyset$, except for $\{a,b,x\}$ and $\{a,b,y\}$. For each of these ballots, solve the following linear program:
	\begin{align*}
		\text{maximize } & P(A) \: \text{ subject to $P$ satisfying \eqref{program:PAV counterexample}}
	\end{align*}
	For each such $A$, the optimal solution of the program is 0.
	
	To prove (iii), iterate through all $c \in W \setminus \{a,b\}$ and solve the following linear programs:
	\begin{align*}
		\text{maximize }& \PAVscore_P(W \setminus \{c\}) - \PAVscore_P(W) \: \text{ subject to $P$ satisfying \eqref{program:PAV counterexample}} \\
		\text{minimize }& \PAVscore_P(W \setminus \{c\}) - \PAVscore_P(W) \: \text{ subject to $P$ satisfying \eqref{program:PAV counterexample}}
	\end{align*}
	The optimal solutions to these two programs are $-1/12$.
\end{proof}

The claims made in this proof about the optimal values of the various linear programs can be certified by exhibiting solutions to the dual programs. These certificates (using exact fractions, not floating point numbers) are \href{https://github.com/DominikPeters/core-few-candidates/tree/master/B_committee_size_8}{available on GitHub}, together with a script checking their validity without using a solver.

As we discussed, \Cref{ex:pav-k8} shows an example of a global PAV committee that is not core-stable, but there are other global PAV committees for the same profile that are core-stable. Thanks to \Cref{lem:pav-k8-counterexamples}, we deduce that the same holds for \emph{all} counterexamples. Hence, for every instance, at least one global PAV committee is in the core, and thus the core is always non-empty for $k = 8$.

\begin{theorem}
	\label{thm:pav-k8}
	\!When $k = 8$, some global PAV committee is in the core.
\end{theorem}
\begin{proof}
	Let $k = 8$ and let $P$ be a profile. If on $P$, all global PAV committees are in the core, we are done. So suppose that $W$ is a global PAV committee that is not core stable due to objection $T$.
	From \Cref{lem:pav-k8-counterexamples}, there exist distinct $a,b \in W$ and distinct $x, y \in C \setminus W$ such that $T = \{a,b,x,y\}$. Take any $c \in W \setminus \{a, b\}$. Then the committee $W' = W \setminus \{c\} \cup \{x\}$ has the same PAV score as $W$, because the removal of $c$ causes a decrease in PAV score of $1/12$ and the addition of $x$ causes an increase of at least $\frac{1}{4} \cdot \frac{1}{3} = 1/12$ due to the quarter of voters from (i) with ballots $A$ such that $A \cap (W \cup T) = \{a, b, x\}$. Thus, $W'$ is also a global PAV committee. We now show that $W'$ is in the core. 
	
	If not, we can apply \Cref{lem:pav-k8-counterexamples} to $W'$ which gives us an objection $T' = \{a', b', x', y'\}$ to $W'$. Clearly, voters with ballots such that $A \cap (W \cup T) = \{a,b,x\}$ are not part of a blocking coalition because $\{a,b,x\} \subseteq W'$. Thus, we deduce that $a,b \not\in T$ from (ii). Thus, the voters with ballots such that $A \cap (W \cup T) = \{a,b,y\}$ are also not supporters of $T'$. Then from part (i) we deduce that the only members of $W$ that are approved by any voters in $P$ are $a$, $b$, $a'$, and $b'$. Thus, there exists a member of $W \setminus \{a, b\}$ who is not approved by any voter, so the removal of that member does not lead to a reduction in PAV score, contradicting (iii).
\end{proof}

\subsection{Committee size $k \ge 9$}
The PAV-based technique that worked for up to $k = 8$ does not continue to work for $k = 9$, since there are examples where there is a unique global PAV committee which fails to be in the core. The following example has this property, and it is the smallest such example with respect to the number of voters ($n = 27$). The indicated committee is also the unique local PAV committee.

{\setlength{\abovedisplayskip}{0pt}\setlength{\belowdisplayskip}{0pt}
\[
\begin{tikzpicture}
	[yscale=0.4,xscale=0.45,voter/.style={anchor=south, yshift=-7pt}, select/.style={fill=blue!10}, c/.style={anchor=south, yshift=1.5pt, inner sep=0}]
	\draw[select] (0,0) rectangle (12,1);
	\draw[select] (0,1) rectangle (12,2);
	\draw (0,2) rectangle (6,3);
	\draw (6,2) rectangle (12,3);
	
	\draw[select] (12,0) rectangle (27,1);
	\draw[select] (12,1) rectangle (27,2);
	\draw[select] (12,2) rectangle (27,3);
	\draw[select] (12,3) rectangle (27,4);
	\draw[select] (12,4) rectangle (27,5);
	\draw[select] (12,5) rectangle (27,6);
	\draw[select] (12,6) rectangle (27,7);

	% Candidate labels
	\node at (6,0.42) {$c_1$};
	\node at (6,1.42) {$c_2$};
	\node at (3,2.42) {$c_3$};
	\node at (9,2.42) {$c_4$};
	\node at (19.5,0.42) {$c_5$};
	\node at (19.5,1.42) {$c_6$};
	\node at (19.5,2.42) {$c_7$};
	\node at (19.5,3.42) {$c_8$};
	\node at (19.5,4.42) {$c_9$};
	\node at (19.5,5.42) {$c_{10}$};
	\node at (19.5,6.42) {$c_{11}$};
	
	\foreach \i in {1,27}
	\node[voter] at (\i-0.5,-1) {$v_{\i}$};
	\node[voter] at (5,-1) {$v_{6}$};
	\node[voter] at (7,-1) {$v_{7}$};
	\node[voter] at (11,-1) {$v_{12}$};
	\node[voter] at (13,-1) {$v_{13}$};

	\node[voter] at (3,-1) {$\cdots$};
	\node[voter] at (9,-1) {$\cdots$};
	\node[voter] at (19.5,-1) {$\cdots$};
\end{tikzpicture}
\]
}

\noindent
\citet[Example 6]{ejr} gave an example where PAV uniquely selects a non-core-stable committee for $k = 10$ and $n = 20$.

%Non-unique committee but even with three voters:
%\[
%\begin{tikzpicture}
%	[yscale=0.43,xscale=0.78,voter/.style={anchor=south, yshift=-7pt}, select/.style={fill=blue!10}, c/.style={anchor=south, yshift=1.5pt, inner sep=0}]
%	\draw[select] (0,0) rectangle (2,1);
%	\draw[select] (0,1) rectangle (2,2);
%	\draw[select] (0,2) rectangle (2,3);
%	\draw (0,3) rectangle (1,4);
%	\draw (1,3) rectangle (2,4);
%	
%	\draw[select] (2,0) rectangle (3,1);
%	\draw[select] (2,1) rectangle (3,2);
%	\draw[select] (2,2) rectangle (3,3);
%	\draw[select] (2,3) rectangle (3,4);
%	\draw[select] (2,4) rectangle (3,5);
%	\draw[select] (2,5) rectangle (3,6);
%	
%	% Candidate labels
%	\node at (1,0.42) {$c_1$};
%	\node at (1,1.42) {$c_2$};
%	\node at (1,2.42) {$c_3$};
%	\node at (0.5,3.42) {$c_4$};
%	\node at (1.5,3.42) {$c_5$};
%	\node at (2.5,0.42) {$c_6$};
%	\node at (2.5,1.42) {$c_7$};
%	\node at (2.5,2.42) {$c_8$};
%	\node at (2.5,3.42) {$c_9$};
%	\node at (2.5,4.42) {$c_{10}$};
%	\node at (2.5,5.42) {$c_{11}$};
%	
%	% Voter labels
%	\foreach \i in {1,2,3}
%	\node[voter] at (\i-0.5,-1) {$v_\i$};
%\end{tikzpicture}
%\]

\section{Few Candidates}\label{sec:small-m}

The goal of this section is to show that there always exists a core-stable committee on instances with $m \le 15$ candidates. From the results in \Cref{sec:small-k}, this is clearly true when $k \le 8$. By \Cref{lem:pav-special-core-deviations}(ii), this is also true when $k = m - 1$ or $k = m$. But it is not clear when $k \in \{9, \dots, m-2\}$.

Inspecting the examples in \Cref{sec:small-k} where PAV fails the core, we see that they are well-structured. Indeed, they are even \emph{laminar instances} in the sense of \citet[Definition 2]{PeSk20}, and it is easy to see that on these profiles, a core-stable committee does exist. Thus, there is some hope to prove existence of core-stable committees by ``patching'' the PAV committee when it fails to be in the core.

We will define an artificial rule, based on PAV, and we will show that it satisfies core stability for up to $m = 15$ candidates. We call it the \emph{recursive PAV rule}. On a high level, the rule first computes a local PAV committee, and checks if it satisfies the core. If so, it returns it. If not, and $T$ is a deviation from $W$, it then deletes all voters who prefer $T$ to $W$, and computes a local PAV committee with respect to the remaining voters, but subject to the constraint that $T \subseteq W$. It then checks if the result is in the core; if not, it adds additional constraints until it reaches a core-stable committee. 
This rule is formally described using pseudocode in \Cref{alg:recursive PAV}.

For example, in the profile of \Cref{ex:pav-k8}, PAV selects the committee $W$ indicated there in blue (among other tied committees). This committee is blocked by $T = \{c_1, c_2, c_3, c_4\}$. Thus, the recursive PAV rule would now fix all members of $T$ as winners, and maximize the PAV-score with respect to voters $v_3$ and $v_4$, obtaining the committee $W^* = \{c_1, c_2, c_3, c_4, c_5, c_6, c_7, c_8\}$ (among other tied committees). This committee is in the core, so the rule terminates. Note that $W^*$ is not itself selected by PAV, though as we saw previously, there do exist other (global) PAV committees in the core in \Cref{ex:pav-k8}.

\newcommand{\mathtextover}[3][l]{\mathmakebox[\widthof{\(#3\)}][#1]{#2}}
\begin{algorithm}[t]
	\begin{algorithmic}
		\State \textbf{Input}: A profile $P$ and a committee size $k$
		\State \textbf{Output}: A committee $W$
		\State $\A' \gets \A$, set of active ballots
		\State $\mathtextover[c]{F}{\A'}\gets \emptyset$, set of \emph{fixed} candidates
		\While{true}
			\State If $|F| > k$, the algorithm \textbf{fails}
			\State $W \gets$ any committee locally maximizing the PAV score 
			\Statex \phantom{$W \gets$} w.r.t.\ the ballots in $\A'$ and subject to $F \subseteq W$
			\If{there exists a successful deviation $T$ from $W$}
				\State $\mathtextover[c]{F}{\A'} \gets F \cup T$
				\State $\A' \gets \A' \setminus \{ A \in \A : u_A(T) > u_A(W) \}$
			\Else
				\State \Return $W$
			\EndIf
		\EndWhile
	\end{algorithmic}
	\caption{Recursive PAV rule}
	\label{alg:recursive PAV}
\end{algorithm}

This method is reminiscent of the Greedy Cohesive Rule \citep{peters2021proportional}, which similarly repeatedly patches a committee until it satisfies the representation axiom FJR.

\subsection{Analysis of the Method}

Fix a number of candidates $m$ and a committee size $k$. 

A list $(W_1, T_1), (W_2, T_2), \dots, (W_r, T_r)$ is called a \emph{potential history} if for each $t \in [r]$, we have that $W_t$ is a committee, $T_t$ is a potential deviation, and $T_1 \cup \dots \cup T_{t-1} \subseteq W_t$.

\begin{definition}
	\label{def:history}
	A potential history $(W_1, T_1), \dots, (W_r, T_r)$ is a \emph{history} if there exists a profile $P$ such that for each $t \in [r]$ we have that $T_t$ is a \emph{successful} deviation from $W_t$, and that for all 
	$x \in W \setminus (T_1 \cup \dots \cup T_{t-1})$ and $y \in C \setminus W$, we have
	{\setlength{\belowdisplayskip}{3pt}
	\[
		 \sum_{A \in \A_t} P(A) \cdot \PAVscore_A(W_t)
		 \ge
		 \sum_{A \in \A_t} P(A) \cdot \PAVscore_A(W_{xy})
	\]}
	where $\A_t = \{ A \in \A : u_A(T_s) \le u_A(W_s) \text{ for } s \in \{1, \dots, t-1\} \}$ is the set of ``active'' ballots. That is, $W_t$ locally maximizes the PAV score among all committees that include all prior deviations, taking only those voters into account that did not participate in prior deviations.
\end{definition}

Thus, a history provides a trace of the execution of \Cref{alg:recursive PAV} for some profile. The following result states that it is enough to analyze the set of histories to determine if \Cref{alg:recursive PAV} always terminates with a core-stable committee.

\begin{proposition}
	\label{prop:limited history size}
	Suppose that for every history $(W_1, T_1), \dots, (W_r, T_r)$, we have $|T_1| + \dots + |T_r| \le k$. Then a core-stable committee always exists for $m$ and $k$.
\end{proposition}
\begin{proof}
	Let $P$ be a profile, and run \Cref{alg:recursive PAV} on it. By the assumption, in each iteration, $|F| \le |T_1| + \dots + |T_r| \le k$, so the algorithm does not fail. By the if-clause, if the algorithm terminates, it returns a committee that is core-stable. Thus, it suffices to show that the algorithm terminates.
	
	Note that after each iteration of the algorithm, it either terminates or it has found a successful deviation. Suppose iteration $r$ has ended without the algorithm terminating. The sequence of committees and deviations $(W_1, T_1), \dots, (W_r, T_r)$ identified by the algorithm up to iteration $r$ forms a history. Since $|T_t| \ge 1$ for all $t$, it follows from $|T_1| + \dots + |T_r| \le k$ that $r \le k$. So it must terminate after at most $k+1$ iterations.
\end{proof}

\begin{algorithm}[t]
	\begin{algorithmic}
		\State \textbf{Input}: Number $m$ of candidates and a committee size $k$
		\State \textbf{Output}: A collection of all histories and Farkas certificates
		\State $\mathcal{H}_0 \gets \{\emptyset\}$, the empty history
		\For{$t = 1, 2, \dots$}
			\ForAll{$H \in \mathcal{H}_{t-1}$}
				\ForAll{potential continuations $(W_t, T_t)$}
					\If{$(W_t, T_t)$ is not canonical with respect to $H$}
						\State \textbf{continue}
					\EndIf
					\State Set $H' \gets H + (W_t, T_t)$
					\State Solve LP to check if $H'$ is a history
					\State If yes, add $H'$ to $\mathcal{H}_t$
					\State If no, generate a Farkas certificate
				\EndFor
			\EndFor
			\If{$\mathcal{H}_t = \emptyset$}
				\State{\textbf{break}}
			\EndIf
		\EndFor
	\end{algorithmic}
	\caption{Finding all histories}
	\label{alg:find all histories}
\end{algorithm}

Thus, to prove the existence of core-stable committees, it suffices to enumerate all histories and check that they fulfil the condition of \Cref{prop:limited history size}.
Given a potential history, one can check using an LP solver whether it is a history by checking whether the system of linear inequalities in \Cref{def:history} has a solution.
This way, we can compute the set of histories using a standard breadth-first search, as shown in \Cref{alg:find all histories}.
A key insight to speed up the search is that we may break symmetries and only need to consider ``canonical'' histories in our enumeration. For example, we may assume without loss of generality that $W_1$, the first committee of the history, is $\{c_1, \dots, c_k\}$. Similarly, we do not need to consider all potential deviations $T_1$: given our choice of $W_1$, the candidates in $W_1$ are indistinguishable to each other, as are the candidates in $C \setminus W_1$, and thus it suffices to take one deviation for each possible combination of the sizes of $|T_1 \cap W_1|$ and of $|T_1 \cap  (C \setminus W_1)|$. Similar symmetry-breaking conditions apply for later steps.%
\footnote{Abstractly speaking, if an algorithm has identified sets $A_1, \dots, A_s \subseteq C$ thus far (where in the present application, these sets are $W_t$'s and $T_t$'s), we can define an equivalence relation with $a \sim b$ if and only if for each $t \in [s]$, either $a,b \in A_t$ or $a,b \in C \setminus A_t$. Then two possibilities $A_{s+1}$ and $A_{s+1}'$ for the next set in the sequence are indistinguishable if for every equivalence class of $\sim$, the two possibilities contain the same number of elements from that equivalence class. A \emph{canonical} choice for the next set would choose the lexicographically first elements from each equivalence class.}

For example, for $m = 15$ and $k = 13$, \Cref{alg:find all histories} produces the following set of (canonical) histories, where we write $W_1 = \{c_{1}, \dots, c_{13}\}$ and $W_2 = \{c_{1}, \dots, c_{11}, c_{14}, c_{15}\}$.
{\addtolength{\jot}{-0.6pt}\begin{align*}
		& \emptyset \text{, the empty history} \\[-1.3pt]
		& (W_1, \{c_{1}, c_{2}, c_{3}, c_{4}, c_{5}, c_{6}, c_{7}, c_{8}, c_{14}, c_{15}\}) \\
		& (W_1, \{c_{1}, c_{14}, c_{15}\}) \\
		& (W_1, \{c_{1}, c_{2}, c_{3}, c_{14}, c_{15}\}) \\
		& (W_1, \{c_{1}, c_{2}, c_{14}, c_{15}\}) \\
		& (W_1, \{c_{1}, c_{2}, c_{3}, c_{4}, c_{5}, c_{14}, c_{15}\}) \\
		& (W_1, \{c_{1}, c_{2}, c_{3}, c_{4}, c_{14}, c_{15}\}) \\
		& (W_1, \{c_{1}, c_{2}, c_{3}, c_{4}, c_{5}, c_{6}, c_{7}, c_{14}, c_{15}\}) \\
		& (W_1, \{c_{1}, c_{2}, c_{3}, c_{4}, c_{5}, c_{6}, c_{14}, c_{15}\}) \\
		& \mathmakebox[4.5cm][l]{(W_1, \{c_{1}, c_{14}, c_{15}\}),}               (W_2, \{c_{2}, c_{12}, c_{13}\}) \\
		& \mathmakebox[4.5cm][l]{(W_1, \{c_{1}, c_{14}, c_{15}\}),}               (W_2, \{c_{2}, c_{3}, c_{12}, c_{13}\}) \\
		& \mathmakebox[4.5cm][l]{(W_1, \{c_{1}, c_{2}, c_{3}, c_{14}, c_{15}\}),} (W_2, \{c_{4}, c_{5}, c_{12}, c_{13}\}) \\
		& \mathmakebox[4.5cm][l]{(W_1, \{c_{1}, c_{2}, c_{14}, c_{15}\}),}        (W_2, \{c_{3}, c_{12}, c_{13}\}) \\
		& \mathmakebox[4.5cm][l]{(W_1, \{c_{1}, c_{2}, c_{14}, c_{15}\}),}        (W_2, \{c_{3}, c_{4}, c_{12}, c_{13}\}) \\
		& \mathmakebox[4.5cm][l]{(W_1, \{c_{1}, c_{2}, c_{14}, c_{15}\}),}        (W_2, \{c_{3}, c_{4}, c_{5}, c_{12}, c_{13}\})
\end{align*}}

By running \Cref{alg:find all histories} for $m = 15$ and $k = 9, \dots, 13$, we obtain the following result. (Note that existence for $m = 15$ implies existence for all $m \le 15$.)

\begin{theorem}
	\label{thm:m-15}
	If $m \le 15$, a core-stable committee exists.
\end{theorem}

The computations establishing \Cref{thm:m-15} can be verified based on Farkas certificates: the \href{https://github.com/DominikPeters/core-few-candidates/tree/master/C_recursive_PAV_rule}{code repository} includes, for each history and each possible extension of the history that induces an infeasible system of linear inequalities, a Farkas witness. Each witness is a list of about $t \cdot k \cdot (m-k)$ integers, where $t$ is the length of the history, corresponding to the constraints in \Cref{def:history}, and verifying the correctness of the witness requires checking about $2^m$ inequalities. In total, there are 114\,373 witnesses (taking 125 MB) and verifying their validity using a simple script \href{https://github.com/DominikPeters/core-few-candidates/tree/master/C_recursive_PAV_rule}{[GitHub]} performing exact fractional computations (without calling a solver) takes about 4 hours on 8 cores (see \Cref{tbl:stats}).

\begin{table}[t]
	\centering
	\begin{tabular}{lrrrrr}
		\toprule
		\hspace{3.5cm} $k = $ & 9 & 10 & 11 & 12 & 13 \\
		\midrule
		number of canonical histories & 7 & 11 & 15 & 20 & 15 \\
		number of Farkas witnesses & 20\,476 & 25\,313 & 18\,567 & 43\,140 & 6\,877  \\
		time for checking Farkas (s) & 2\,648 & 3\,301 & 2\,087 & 5\,857 & 725 \\
		\bottomrule
	\end{tabular}
	\caption{Statistics about the histories for $m = 15$.}
	\label{tbl:stats}
\end{table}

%\dominik{Note somewhere that this rule is unfortunately not nicely behaved, and from the second step the committee can even fail JR (singleton $T$).}

The recursive PAV rule fails for $m = 16$, $k \in \{10, 11\}$. For $m = 16$, $k = 10$, the smallest failure example I have found has 40\,448\,550 voters, though the ILP for minimizing this number did not converge within a reasonable amount of time. The example is \href{https://github.com/DominikPeters/core-few-candidates/blob/master/C_recursive_PAV_rule/counterexample-m16-k10.json}{available online}.%
\footnote{The history that witnesses this failure is the following: $(\{c_{0}, c_{1}, c_{2}, c_{3}, c_{4}, c_{5}, c_{6}, c_{7}, c_{8}, c_{9}\}, \{c_{0}, c_{10}, c_{11}\}),$ 
$(\{c_{0}, c_{1}, c_{2}, c_{3}, c_{4}, c_{5}, c_{6}, c_{10}, c_{11}, c_{12}\}, \{c_{13}, c_{14}, c_{15}\}),$ 
$(\{c_{0}, c_{1}, c_{2}, c_{3}, c_{10}, c_{11}, c_{12}, c_{13}, c_{14}, c_{15}\}, \{c_{4}, c_{5}, c_{6}, c_{7}, c_{8}\})$.}
The recursive PAV rule \emph{does} work for $m = 16$, $k \in \{9, 12, 13, 14\}$, and it is plausible that it can be fixed \textit{ad~hoc} for $k \in \{10, 11\}$, so it is likely that the core continues to exist for $m = 16$.

\section{Droop Quota}\label{sec:droop}
Our definition of core stability is based on the intuition that a $1/k$ fraction of the voters is ``entitled'' to decide on one of the committee members, and that an $\ell/k$ fraction is entitled to decide on $\ell$ committee members. The quantity $1/k$ is known as the \emph{Hare quota}. But one can also define core stability based on the \emph{Droop quota}, according to which each group of voters that makes up a strictly larger fraction than $1/(k+1)$ is entitled to decide on one committee member. Thus, a committee $W$ is \emph{Droop core stable} if for every potential deviation $T$, we have 
\[
\sum_{A \in \A : u_A(T) > u_A(W)} P(A) \le \frac{|T|}{k+1}.
\]
This is a stricter condition than the normal core, so if $W$ is Droop core stable then it is also core stable.

For most proportionality notions considered in the literature on approval-based committee elections, passing to the more demanding Droop quota does not cause many issues. For example, PAV still satisfies EJR when defined with the Droop quota, and analogous statements are true for many pairs of voting rules and representation axioms \citep{janson2018thresholds}.%
\footnote{However, regarding strategic aspects, impossibility theorems become somewhat more expansive when passing to the Droop quota \citep[Section 5.3]{peters2018proportionality}.}

Unfortunately, our positive results do not extend to the Droop core. While PAV satisfies the core for up to $k = 8$ (\Cref{sec:small-k}), it violates the Droop core already for $k = 6$.

\begin{example}[PAV may fail the Droop core for $k = 6$]
	\label{ex:pav-droop-k6}
	Consider the instance depicted below:
	\[
	\begin{tikzpicture}
		%		xscale=0.62
		[yscale=0.4,xscale=0.4,voter/.style={anchor=south, yshift=-7pt}, select/.style={fill=blue!10}, c/.style={anchor=south, yshift=1.5pt, inner sep=0}]
		\draw[select] (0,0) rectangle (14,1);
		\draw[select] (0,1) rectangle (14,2);
		\draw (0,2) rectangle (7,3);
		\draw (7,2) rectangle (14,3);
		
		\draw[select] (14,0) rectangle (24,1);
		\draw[select] (14,1) rectangle (24,2);
		\draw[select] (14,2) rectangle (24,3);
		\draw[select] (14,3) rectangle (24,4);

		% Candidate labels
		\node at (7,0.42) {$c_1$};
		\node at (7,1.42) {$c_2$};
		\node at (3.5,2.42) {$c_3$};
		\node at (10.5,2.42) {$c_4$};
		\node at (19,0.42) {$c_5$};
		\node at (19,1.42) {$c_6$};
		\node at (19,2.42) {$c_7$};
		\node at (19,3.42) {$c_8$};
		
		\foreach \i in {1,24}
			\node[voter] at (\i-0.5,-1) {$v_{\i}$};
		\node[voter] at (6,-1) {$v_{7}$};
		\node[voter] at (8,-1) {$v_{8}$};
		\node[voter] at (13,-1) {$v_{14}$};
		\node[voter] at (15,-1) {$v_{15}$};
	
	\node[voter] at (3.5,-1) {$\cdots$};
	\node[voter] at (10.5,-1) {$\cdots$};
	\node[voter] at (19,-1) {$\cdots$};
	\end{tikzpicture}
	\]
	On this profile, $W = \{c_1, c_2, c_5, c_6, c_7, c_8 \}$ is the unique global (and unique local) PAV committee for $k = 6$. However $W$ is not in the Droop core: consider $T = \{c_1, c_2, c_3, c_4\}$, which has support from $\frac{14}{24} \approx 0.583$ of the voters, while $|T|/(k+1) = \frac{4}{7} \approx 0.571$ is strictly smaller.
	\qed
\end{example}
This example is minimal, so the Droop core is non-empty when $k \le 5$.
Running the recursive PAV rule (\Cref{alg:recursive PAV}) with the Droop quota stops working even for $m = 10$, $k = 6$.

\section{Conclusions}\label{sec:conclusions}

Based on the computations of this paper, we know that the core is non-empty for all small instances. This should probably strengthen our belief that the core is always non-empty. However, the recursive PAV method we defined to establish the result stops working for $16$ or more candidates, so it seems doubtful that analyzing this method would allow proving a general existence result. Conversely, finding a counterexample to core existence will also be challenging since it will need to be large. For the Droop quota, however, it even remains unknown whether core always exists for $k = 6$ and $m = 10$.

Our approach was based on linear programming, and in particular this approach allowed us to reason independently of the number of voters. The PAV rule and its variants are particularly well-suited for these LP formulations. However, finding core counterexamples for many other rules is not possible using similar linear programs. For example, the Method of Equal Shares (MES) \citep{PeSk20} or the sequential Phragm\'en method \citep{Phra94a,Jans16a} do not admit the same kind of linear formulations (because they would require multiplying variables corresponding to ballot frequencies with variables corresponding to $\rho$-values or to loads). The lack of such a linear formulation can be formally established using the techniques of \citet{xia2025lineartheory}. In part because computer search is difficult for these rules, to the best of my knowledge, there is no known profile where both PAV and MES fail core-stability simultaneously. I am also not aware of any example where the rule that maximizes the PAV score among all \emph{priceable} committees \citep{PeSk20} fails core-stability.

\citet{maly2023coreapprovalbasedpbinstance} presents an example in the participatory budgeting setting with cost utilities where the core is empty. That example uses only 3 voters. It would be interesting to see if computer-aided methods could establish that for committee elections, the core is always non-empty for $n = 3$ voters. Note that in this case, candidates can be specified via the set of voters that approve the candidate, so there are only $2^3$ different types of candidates, and thus a profile can be specified via variables that indicate how many candidates of each type exist.

\section*{Acknowledgements}
I thank Paul Gölz for useful discussions, and Jannik Peters and the anonymous reviewers at IJCAI 2025 for feedback that improved the presentation of the paper. I have used Gurobi and the cvc5 solver \citep{cvc5} in this work. This work was funded in part by the Agence Nationale de la Recherche under grant ANR22-CE26-0019 (CITIZENS) and as part of the France 2030 program under grant ANR-23-IACL-0008 (PR[AI]RIE-PSAI).

%\bibliography{core-few-candidates}

\begin{thebibliography}{36}
	
	%%% the Web address.
	
	\ifx \showCODEN    \undefined \def \showCODEN     #1{\unskip}     \fi
	\ifx \showISBNx    \undefined \def \showISBNx     #1{\unskip}     \fi
	\ifx \showISBNxiii \undefined \def \showISBNxiii  #1{\unskip}     \fi
	\ifx \showISSN     \undefined \def \showISSN      #1{\unskip}     \fi
	\ifx \showLCCN     \undefined \def \showLCCN      #1{\unskip}     \fi
	\ifx \shownote     \undefined \def \shownote      #1{#1}          \fi
	\ifx \showarticletitle \undefined \def \showarticletitle #1{#1}   \fi
	\ifx \showURL      \undefined \def \showURL       {\relax}        \fi
	% invisible to TeX
	\providecommand\bibfield[2]{#2}
	\providecommand\bibinfo[2]{#2}
	\providecommand\natexlab[1]{#1}
	\providecommand\showeprint[2][]{arXiv:#2}
	
	\bibitem[Aziz et~al\mbox{.}(2020)]%
	{ABM20a}
	\bibfield{author}{\bibinfo{person}{Haris Aziz}, \bibinfo{person}{Anna
			Bogomolnaia}, {and} \bibinfo{person}{Herv\'{e} Moulin}.}
	\bibinfo{year}{2020}\natexlab{}.
	\newblock \showarticletitle{Fair mixing: the case of dichotomous preferences}.
	\newblock \bibinfo{journal}{\emph{ACM Transactions on Economics and
			Computation}} \bibinfo{volume}{8}, \bibinfo{number}{4}
	(\bibinfo{year}{2020}), \bibinfo{pages}{18:1--18:27}.
	\newblock
	\href{https://doi.org/10.1145/3417738}{doi:\nolinkurl{10.1145/3417738}}
	
	
	\bibitem[Aziz et~al\mbox{.}(2016)]%
	{ejr-v4}
	\bibfield{author}{\bibinfo{person}{Haris Aziz}, \bibinfo{person}{Markus Brill},
		\bibinfo{person}{Vincent Conitzer}, \bibinfo{person}{Edith Elkind},
		\bibinfo{person}{Rupert Freeman}, {and} \bibinfo{person}{Toby Walsh}.}
	\bibinfo{year}{2016}\natexlab{}.
	\newblock \bibinfo{title}{Justified representation in approval-based committee
		voting}.
	\newblock
	\showeprint[arxiv]{1407.8269v4}~[cs.MA]
	\urldef\tempurl%
	\url{https://arxiv.org/abs/1407.8269v4}
	\showURL{%
		\tempurl}
	
	
	\bibitem[Aziz et~al\mbox{.}(2017)]%
	{ejr}
	\bibfield{author}{\bibinfo{person}{Haris Aziz}, \bibinfo{person}{Markus Brill},
		\bibinfo{person}{Vincent Conitzer}, \bibinfo{person}{Edith Elkind},
		\bibinfo{person}{Rupert Freeman}, {and} \bibinfo{person}{Toby Walsh}.}
	\bibinfo{year}{2017}\natexlab{}.
	\newblock \showarticletitle{Justified representation in approval-based
		committee voting}.
	\newblock \bibinfo{journal}{\emph{Social Choice and Welfare}}
	\bibinfo{volume}{48}, \bibinfo{number}{2} (\bibinfo{year}{2017}),
	\bibinfo{pages}{461--485}.
	\newblock
	\href{https://doi.org/10.1007/s00355-016-1019-3}{doi:\nolinkurl{10.1007/s00355-016-1019-3}}
	
	
	\bibitem[Aziz et~al\mbox{.}(2018)]%
	{aziz2018complexity}
	\bibfield{author}{\bibinfo{person}{Haris Aziz}, \bibinfo{person}{Edith Elkind},
		\bibinfo{person}{Shenwei Huang}, \bibinfo{person}{Martin Lackner},
		\bibinfo{person}{Luis S{\'a}nchez-Fern{\'a}ndez}, {and}
		\bibinfo{person}{Piotr Skowron}.} \bibinfo{year}{2018}\natexlab{}.
	\newblock \showarticletitle{On the complexity of extended and proportional
		justified representation}. In \bibinfo{booktitle}{\emph{Proceedings of the
			32nd AAAI Conference on Artificial Intelligence (AAAI)}}.
	\bibinfo{pages}{902--909}.
	\newblock
	\href{https://doi.org/10.1609/aaai.v32i1.11478}{doi:\nolinkurl{10.1609/aaai.v32i1.11478}}
	
	
	\bibitem[Aziz et~al\mbox{.}(2015)]%
	{aziz2015computational}
	\bibfield{author}{\bibinfo{person}{Haris Aziz}, \bibinfo{person}{Serge
			Gaspers}, \bibinfo{person}{Joachim Gudmundsson}, \bibinfo{person}{Simon
			Mackenzie}, \bibinfo{person}{Nicholas Mattei}, {and} \bibinfo{person}{Toby
			Walsh}.} \bibinfo{year}{2015}\natexlab{}.
	\newblock \showarticletitle{Computational aspects of multi-winner approval
		voting}. In \bibinfo{booktitle}{\emph{Proceedings of the 14th International
			Conference on Autonomous Agents and Multiagent Systems (AAMAS)}}.
	\bibinfo{pages}{107--115}.
	\newblock
	\urldef\tempurl%
	\url{https://www.ifaamas.org/Proceedings/aamas2015/aamas/p107.pdf}
	\showURL{%
		\tempurl}
	
	
	\bibitem[Barbosa et~al\mbox{.}(2022)]%
	{cvc5}
	\bibfield{author}{\bibinfo{person}{Haniel Barbosa}, \bibinfo{person}{Clark~W.
			Barrett}, \bibinfo{person}{Martin Brain}, \bibinfo{person}{Gereon Kremer},
		\bibinfo{person}{Hanna Lachnitt}, \bibinfo{person}{Makai Mann},
		\bibinfo{person}{Abdalrhman Mohamed}, \bibinfo{person}{Mudathir Mohamed},
		\bibinfo{person}{Aina Niemetz}, \bibinfo{person}{Andres N{\"{o}}tzli},
		\bibinfo{person}{Alex Ozdemir}, \bibinfo{person}{Mathias Preiner},
		\bibinfo{person}{Andrew Reynolds}, \bibinfo{person}{Ying Sheng},
		\bibinfo{person}{Cesare Tinelli}, {and} \bibinfo{person}{Yoni Zohar}.}
	\bibinfo{year}{2022}\natexlab{}.
	\newblock \showarticletitle{cvc5: {A} versatile and industrial-strength {SMT}
		solver}. In \bibinfo{booktitle}{\emph{Proceedings of the 28th International
			Conference on Tools and Algorithms for the Construction and Analysis of
			Systems (TACAS)}} \emph{(\bibinfo{series}{Lecture Notes in Computer Science},
		Vol.~\bibinfo{volume}{13243})}. \bibinfo{publisher}{Springer},
	\bibinfo{pages}{415--442}.
	\newblock
	\href{https://doi.org/10.1007/978-3-030-99524-9\_24}{doi:\nolinkurl{10.1007/978-3-030-99524-9\_24}}
	
	
	\bibitem[Bei et~al\mbox{.}(2025)]%
	{bei2024truthful}
	\bibfield{author}{\bibinfo{person}{Xiaohui Bei}, \bibinfo{person}{Xinhang Lu},
		{and} \bibinfo{person}{Warut Suksompong}.} \bibinfo{year}{2025}\natexlab{}.
	\newblock \showarticletitle{Truthful cake sharing}.
	\newblock \bibinfo{journal}{\emph{Social Choice and Welfare}}
	\bibinfo{volume}{64} (\bibinfo{year}{2025}), \bibinfo{pages}{309--343}.
	\newblock
	\href{https://doi.org/10.1007/s00355-023-01503-0}{doi:\nolinkurl{10.1007/s00355-023-01503-0}}
	
	
	\bibitem[Brandl et~al\mbox{.}(2021)]%
	{brandl2021distribution}
	\bibfield{author}{\bibinfo{person}{Florian Brandl}, \bibinfo{person}{Felix
			Brandt}, \bibinfo{person}{Dominik Peters}, {and} \bibinfo{person}{Christian
			Stricker}.} \bibinfo{year}{2021}\natexlab{}.
	\newblock \showarticletitle{Distribution rules under dichotomous preferences:
		Two out of three ain't bad}. In \bibinfo{booktitle}{\emph{Proceedings of the
			22nd ACM Conference on Economics and Computation (EC)}}.
	\bibinfo{pages}{158--179}.
	\newblock
	\href{https://doi.org/10.1145/3465456.3467653}{doi:\nolinkurl{10.1145/3465456.3467653}}
	
	
	\bibitem[Brill et~al\mbox{.}(2023)]%
	{brill2023phragmen}
	\bibfield{author}{\bibinfo{person}{Markus Brill}, \bibinfo{person}{Rupert
			Freeman}, \bibinfo{person}{Svante Janson}, {and} \bibinfo{person}{Martin
			Lackner}.} \bibinfo{year}{2023}\natexlab{}.
	\newblock \showarticletitle{Phragm{\'e}n's voting methods and justified
		representation}.
	\newblock \bibinfo{journal}{\emph{Mathematical Programming}}
	(\bibinfo{year}{2023}).
	\newblock
	\href{https://doi.org/10.1007/s10107-023-01926-8}{doi:\nolinkurl{10.1007/s10107-023-01926-8}}
	
	
	\bibitem[Brill et~al\mbox{.}(2022)]%
	{approvalbasedapportionment}
	\bibfield{author}{\bibinfo{person}{Markus Brill}, \bibinfo{person}{Paul
			G{\"o}lz}, \bibinfo{person}{Dominik Peters}, \bibinfo{person}{Ulrike
			Schmidt-Kraepelin}, {and} \bibinfo{person}{Kai Wilker}.}
	\bibinfo{year}{2022}\natexlab{}.
	\newblock \showarticletitle{Approval-based apportionment}.
	\newblock \bibinfo{journal}{\emph{Mathematical Programming}}
	(\bibinfo{year}{2022}).
	\newblock
	\href{https://doi.org/10.1007/s10107-022-01852-1}{doi:\nolinkurl{10.1007/s10107-022-01852-1}}
	
	
	\bibitem[Brill and Peters(2023)]%
	{ejr+}
	\bibfield{author}{\bibinfo{person}{Markus Brill} {and} \bibinfo{person}{Jannik
			Peters}.} \bibinfo{year}{2023}\natexlab{}.
	\newblock \showarticletitle{Robust and verifiable proportionality axioms for
		multiwinner voting}. In \bibinfo{booktitle}{\emph{Proceedings of the 24th ACM
			Conference on Economics and Computation (EC)}}. \bibinfo{pages}{301}.
	\newblock
	\href{https://doi.org/10.1145/3580507.3597785}{doi:\nolinkurl{10.1145/3580507.3597785}}
	\newblock
	\shownote{Full version
		\href{https://arxiv.org/abs/2302.01989}{arXiv:2302.01989}}.
	
	
	\bibitem[Charikar et~al\mbox{.}(2024)]%
	{charikar2024candidatessufficewinvoter}
	\bibfield{author}{\bibinfo{person}{Moses Charikar}, \bibinfo{person}{Alexandra
			Lassota}, \bibinfo{person}{Prasanna Ramakrishnan}, \bibinfo{person}{Adrian
			Vetta}, {and} \bibinfo{person}{Kangning Wang}.}
	\bibinfo{year}{2024}\natexlab{}.
	\newblock \bibinfo{title}{Six candidates suffice to win a voter majority}.
	\newblock
	\showeprint[arxiv]{2411.03390}~[cs.GT]
	\newblock
	\shownote{To appear in STOC 2025}.
	
	
	\bibitem[Cheng et~al\mbox{.}(2020)]%
	{cheng2020group}
	\bibfield{author}{\bibinfo{person}{Yu Cheng}, \bibinfo{person}{Zhihao Jiang},
		\bibinfo{person}{Kamesh Munagala}, {and} \bibinfo{person}{Kangning Wang}.}
	\bibinfo{year}{2020}\natexlab{}.
	\newblock \showarticletitle{Group fairness in committee selection}.
	\newblock \bibinfo{journal}{\emph{ACM Transactions on Economics and Computation
			(TEAC)}} \bibinfo{volume}{8}, \bibinfo{number}{4} (\bibinfo{year}{2020}),
	\bibinfo{pages}{1--18}.
	\newblock
	\href{https://doi.org/10.1145/3417750}{doi:\nolinkurl{10.1145/3417750}}
	
	
	\bibitem[Fain et~al\mbox{.}(2016)]%
	{FGM16a}
	\bibfield{author}{\bibinfo{person}{Brandon Fain}, \bibinfo{person}{Ashish
			Goel}, {and} \bibinfo{person}{Kamesh Munagala}.}
	\bibinfo{year}{2016}\natexlab{}.
	\newblock \showarticletitle{The core of the participatory budgeting problem}.
	In \bibinfo{booktitle}{\emph{Proceedings of the 12th International Conference
			on Web and Internet Economics (WINE)}}. \bibinfo{pages}{384--399}.
	\newblock
	\href{https://doi.org/10.1007/978-3-662-54110-4_27}{doi:\nolinkurl{10.1007/978-3-662-54110-4_27}}
	
	
	\bibitem[Fain et~al\mbox{.}(2018)]%
	{fain2018fair}
	\bibfield{author}{\bibinfo{person}{Brandon Fain}, \bibinfo{person}{Kamesh
			Munagala}, {and} \bibinfo{person}{Nisarg Shah}.}
	\bibinfo{year}{2018}\natexlab{}.
	\newblock \showarticletitle{Fair allocation of indivisible public goods}. In
	\bibinfo{booktitle}{\emph{Proceedings of the 2018 ACM Conference on Economics
			and Computation (EC)}}. \bibinfo{pages}{575--592}.
	\newblock
	\href{https://doi.org/10.1145/3219166.3219174}{doi:\nolinkurl{10.1145/3219166.3219174}}
	
	
	\bibitem[Foley(1970)]%
	{foley1970lindahl}
	\bibfield{author}{\bibinfo{person}{Duncan~K. Foley}.}
	\bibinfo{year}{1970}\natexlab{}.
	\newblock \showarticletitle{Lindahl's solution and the core of an economy with
		public goods}.
	\newblock \bibinfo{journal}{\emph{Econometrica}} \bibinfo{volume}{38},
	\bibinfo{number}{1} (\bibinfo{year}{1970}), \bibinfo{pages}{66--72}.
	\newblock
	\href{https://doi.org/10.2307/1909241}{doi:\nolinkurl{10.2307/1909241}}
	
	
	\bibitem[Janson(2016)]%
	{Jans16a}
	\bibfield{author}{\bibinfo{person}{Svante Janson}.}
	\bibinfo{year}{2016}\natexlab{}.
	\newblock \bibinfo{title}{Phragm\'en's and {Thiele's} election methods}.
	\newblock
	\showeprint[arxiv]{1611.08826}~[math.HO]
	\urldef\tempurl%
	\url{https://arxiv.org/abs/1611.0882}
	\showURL{%
		\tempurl}
	
	
	\bibitem[Janson(2018)]%
	{janson2018thresholds}
	\bibfield{author}{\bibinfo{person}{Svante Janson}.}
	\bibinfo{year}{2018}\natexlab{}.
	\newblock \bibinfo{title}{Thresholds quantifying proportionality criteria for
		election methods}.
	\newblock
	\showeprint[arxiv]{1810.06377}~[cs.GT]
	\urldef\tempurl%
	\url{https://arxiv.org/abs/1810.06377}
	\showURL{%
		\tempurl}
	
	
	\bibitem[Jiang et~al\mbox{.}(2020)]%
	{jiang2020approximatelystable}
	\bibfield{author}{\bibinfo{person}{Zhihao Jiang}, \bibinfo{person}{Kamesh
			Munagala}, {and} \bibinfo{person}{Kangning Wang}.}
	\bibinfo{year}{2020}\natexlab{}.
	\newblock \showarticletitle{Approximately stable committee selection}. In
	\bibinfo{booktitle}{\emph{Proceedings of the 52nd Annual ACM SIGACT Symposium
			on Theory of Computing (STOC)}}. \bibinfo{pages}{463–472}.
	\newblock
	\href{https://doi.org/10.1145/3357713.3384238}{doi:\nolinkurl{10.1145/3357713.3384238}}
	
	
	\bibitem[Kraiczy and Elkind(2024)]%
	{kraiczy2024localsearch}
	\bibfield{author}{\bibinfo{person}{Sonja Kraiczy} {and} \bibinfo{person}{Edith
			Elkind}.} \bibinfo{year}{2024}\natexlab{}.
	\newblock \showarticletitle{A lower bound for local search Proportional
		Approval Voting}. In \bibinfo{booktitle}{\emph{Proceedings of the 32nd Annual
			European Symposium on Algorithms (ESA)}}. \bibinfo{pages}{82:1--82:14}.
	\newblock
	\href{https://doi.org/10.4230/LIPIcs.ESA.2024.82}{doi:\nolinkurl{10.4230/LIPIcs.ESA.2024.82}}
	
	
	\bibitem[Kroer and Peters(2025)]%
	{kroer2025lindahl}
	\bibfield{author}{\bibinfo{person}{Christian Kroer} {and}
		\bibinfo{person}{Dominik Peters}.} \bibinfo{year}{2025}\natexlab{}.
	\newblock \bibinfo{title}{Computing {Lindahl} equilibrium for public goods with
		and without funding caps}.
	\newblock
	\showeprint[arxiv]{2503.16414}~[cs.GT]
	\urldef\tempurl%
	\url{https://arxiv.org/abs/2503.16414}
	\showURL{%
		\tempurl}
	
	
	\bibitem[Maly(2023)]%
	{maly2023coreapprovalbasedpbinstance}
	\bibfield{author}{\bibinfo{person}{Jan Maly}.} \bibinfo{year}{2023}\natexlab{}.
	\newblock \bibinfo{title}{The core of an approval-based {PB} instance can be
		empty for nearly all cost-based satisfaction functions and for the share}.
	\newblock
	\showeprint[arxiv]{2311.06132}~[cs.GT]
	\urldef\tempurl%
	\url{https://arxiv.org/abs/2311.06132}
	\showURL{%
		\tempurl}
	
	
	\bibitem[Munagala et~al\mbox{.}(2022a)]%
	{munagala2022auditing}
	\bibfield{author}{\bibinfo{person}{Kamesh Munagala}, \bibinfo{person}{Yiheng
			Shen}, {and} \bibinfo{person}{Kangning Wang}.}
	\bibinfo{year}{2022}\natexlab{a}.
	\newblock \showarticletitle{Auditing for core stability in participatory
		budgeting}. In \bibinfo{booktitle}{\emph{Proceedings of the 18th
			International Conference on Web and Internet Economics (WINE)}}. Springer,
	\bibinfo{pages}{292--310}.
	\newblock
	\href{https://doi.org/10.1007/978-3-031-22832-2_17}{doi:\nolinkurl{10.1007/978-3-031-22832-2_17}}
	
	
	\bibitem[Munagala et~al\mbox{.}(2022b)]%
	{munagala2022coremultilinear}
	\bibfield{author}{\bibinfo{person}{Kamesh Munagala}, \bibinfo{person}{Yiheng
			Shen}, \bibinfo{person}{Kangning Wang}, {and} \bibinfo{person}{Zhiyi Wang}.}
	\bibinfo{year}{2022}\natexlab{b}.
	\newblock \showarticletitle{Approximate core for committee selection via
		multilinear extension and market clearing}. In
	\bibinfo{booktitle}{\emph{Proceedings of the 2022 Annual ACM-SIAM Symposium
			on Discrete Algorithms (SODA)}}. \bibinfo{pages}{2229--2252}.
	\newblock
	\href{https://doi.org/10.1137/1.9781611977073.89}{doi:\nolinkurl{10.1137/1.9781611977073.89}}
	
	
	\bibitem[Peters(2018)]%
	{peters2018proportionality}
	\bibfield{author}{\bibinfo{person}{Dominik Peters}.}
	\bibinfo{year}{2018}\natexlab{}.
	\newblock \showarticletitle{Proportionality and strategyproofness in
		multiwinner elections}. In \bibinfo{booktitle}{\emph{Proceedings of the 17th
			International Conference on Autonomous Agents and Multiagent Systems
			(AAMAS)}}. \bibinfo{pages}{1549--1557}.
	\newblock
	\urldef\tempurl%
	\url{https://arxiv.org/abs/2104.08594}
	\showURL{%
		\tempurl}
	
	
	\bibitem[Peters et~al\mbox{.}(2021)]%
	{peters2021proportional}
	\bibfield{author}{\bibinfo{person}{Dominik Peters}, \bibinfo{person}{Grzegorz
			Pierczy{\'n}ski}, {and} \bibinfo{person}{Piotr Skowron}.}
	\bibinfo{year}{2021}\natexlab{}.
	\newblock \showarticletitle{Proportional participatory budgeting with additive
		utilities}. In \bibinfo{booktitle}{\emph{Advances in Neural Information
			Processing Systems}}, Vol.~\bibinfo{volume}{34}.
	\bibinfo{pages}{12726--12737}.
	\newblock
	\href{https://doi.org/10.48550/arXiv.2008.13276}{doi:\nolinkurl{10.48550/arXiv.2008.13276}}
	
	
	\bibitem[Peters and Skowron(2020)]%
	{PeSk20}
	\bibfield{author}{\bibinfo{person}{Dominik Peters} {and} \bibinfo{person}{Piotr
			Skowron}.} \bibinfo{year}{2020}\natexlab{}.
	\newblock \showarticletitle{Proportionality and the limits of welfarism}. In
	\bibinfo{booktitle}{\emph{Proceedings of the 21st ACM Conference on Economics
			and Computation (EC)}}. \bibinfo{pages}{793--794}.
	\newblock
	\href{https://doi.org/10.1145/3391403.3399465}{doi:\nolinkurl{10.1145/3391403.3399465}}
	\newblock
	\shownote{Full version
		\href{https://arxiv.org/abs/1911.11747}{arXiv:1911.11747}}.
	
	
	\bibitem[Phragm{\'e}n(1894)]%
	{Phra94a}
	\bibfield{author}{\bibinfo{person}{Edvard Phragm{\'e}n}.}
	\bibinfo{year}{1894}\natexlab{}.
	\newblock \showarticletitle{Sur une m{\'e}thode nouvelle pour r{\'e}aliser,
		dans les {\'e}lections, la repr{\'e}sentation proportionnelle des partis}.
	\newblock \bibinfo{journal}{\emph{\"Ofversigt af Kongliga Vetenskaps-Akademiens
			F\"orhandlingar}} \bibinfo{volume}{51}, \bibinfo{number}{3}
	(\bibinfo{year}{1894}), \bibinfo{pages}{133--137}.
	\newblock
	\urldef\tempurl%
	\url{https://dominik-peters.de/archive/phragmen1894.pdf}
	\showURL{%
		\tempurl}
	
	
	\bibitem[Pierczy{\'n}ski and Skowron(2022)]%
	{pierczynski2022restricteddomains}
	\bibfield{author}{\bibinfo{person}{Grzegorz Pierczy{\'n}ski} {and}
		\bibinfo{person}{Piotr Skowron}.} \bibinfo{year}{2022}\natexlab{}.
	\newblock \showarticletitle{Core-stable committees under restricted domains}.
	In \bibinfo{booktitle}{\emph{Proceedings of the 18th International Conference
			on Web and Internet Economics (WINE)}}. Springer, \bibinfo{pages}{311--329}.
	\newblock
	\href{https://doi.org/10.1007/978-3-031-22832-2_18}{doi:\nolinkurl{10.1007/978-3-031-22832-2_18}}
	
	
	\bibitem[Rey and Maly(2023)]%
	{rey2023computationalsocialchoiceindivisible}
	\bibfield{author}{\bibinfo{person}{Simon Rey} {and} \bibinfo{person}{Jan
			Maly}.} \bibinfo{year}{2023}\natexlab{}.
	\newblock \bibinfo{title}{The (computational) social choice take on indivisible
		participatory budgeting}.
	\newblock
	\showeprint[arxiv]{2303.00621}~[cs.GT]
	\urldef\tempurl%
	\url{https://arxiv.org/abs/2303.00621}
	\showURL{%
		\tempurl}
	
	
	\bibitem[S{\'a}nchez-Fern{\'a}ndez et~al\mbox{.}(2017)]%
	{pjr}
	\bibfield{author}{\bibinfo{person}{Luis S{\'a}nchez-Fern{\'a}ndez},
		\bibinfo{person}{Edith Elkind}, \bibinfo{person}{Martin Lackner},
		\bibinfo{person}{Norberto Fern{\'a}ndez}, \bibinfo{person}{Jes{\'u}s
			Fisteus}, \bibinfo{person}{Pablo~Basanta Val}, {and} \bibinfo{person}{Piotr
			Skowron}.} \bibinfo{year}{2017}\natexlab{}.
	\newblock \showarticletitle{Proportional justified representation}. In
	\bibinfo{booktitle}{\emph{Proceedings of the 31st AAAI Conference on
			Artificial Intelligence (AAAI)}}. \bibinfo{pages}{670--676}.
	\newblock
	\href{https://doi.org/10.5555/3298239.3298338}{doi:\nolinkurl{10.5555/3298239.3298338}}
	
	
	\bibitem[Skowron(2021)]%
	{skowron2021proportionality}
	\bibfield{author}{\bibinfo{person}{Piotr Skowron}.}
	\bibinfo{year}{2021}\natexlab{}.
	\newblock \showarticletitle{Proportionality degree of multiwinner rules}. In
	\bibinfo{booktitle}{\emph{Proceedings of the 22nd ACM Conference on Economics
			and Computation (EC)}}. \bibinfo{pages}{820--840}.
	\newblock
	\href{https://doi.org/10.1145/3465456.3467641}{doi:\nolinkurl{10.1145/3465456.3467641}}
	
	
	\bibitem[Suzuki and Vollen(2024)]%
	{suzuki2024maxflow}
	\bibfield{author}{\bibinfo{person}{Mashbat Suzuki} {and}
		\bibinfo{person}{Jeremy Vollen}.} \bibinfo{year}{2024}\natexlab{}.
	\newblock \showarticletitle{Maximum flow is fair: A network flow approach to
		committee voting}. In \bibinfo{booktitle}{\emph{Proceedings of the 25th ACM
			Conference on Economics and Computation (EC)}}. \bibinfo{pages}{964--983}.
	\newblock
	\href{https://doi.org/10.1145/3670865.3673603}{doi:\nolinkurl{10.1145/3670865.3673603}}
	
	
	\bibitem[Thiele(1895)]%
	{Thie95a}
	\bibfield{author}{\bibinfo{person}{Thorvald~N. Thiele}.}
	\bibinfo{year}{1895}\natexlab{}.
	\newblock \showarticletitle{Om Flerfold Valg}.
	\newblock \bibinfo{journal}{\emph{Oversigt over det Kongelige Danske
			Videnskabernes Selskabs Fordhandlinger}} (\bibinfo{year}{1895}).
	\newblock
	\urldef\tempurl%
	\url{https://dominik-peters.de/archive/thiele1895.pdf}
	\showURL{%
		\tempurl}
	
	
	\bibitem[Xia(2025)]%
	{xia2025lineartheory}
	\bibfield{author}{\bibinfo{person}{Lirong Xia}.}
	\bibinfo{year}{2025}\natexlab{}.
	\newblock \bibinfo{title}{A linear theory of multi-winner voting}.
	\newblock
	\showeprint[arxiv]{2503.03082}~[cs.GT]
	\urldef\tempurl%
	\url{https://arxiv.org/abs/2503.03082}
	\showURL{%
		\tempurl}
	
	
	\bibitem[Yang and Wang(2023)]%
	{yang2023parameterized}
	\bibfield{author}{\bibinfo{person}{Yongjie Yang} {and} \bibinfo{person}{Jianxin
			Wang}.} \bibinfo{year}{2023}\natexlab{}.
	\newblock \showarticletitle{Parameterized complexity of multiwinner
		determination: more effort towards fixed-parameter tractability}.
	\newblock \bibinfo{journal}{\emph{Autonomous Agents and Multi-Agent Systems}}
	\bibinfo{volume}{37}, \bibinfo{number}{2} (\bibinfo{year}{2023}),
	\bibinfo{pages}{28}.
	\newblock
	\href{https://doi.org/10.1007/s10458-023-09610-z}{doi:\nolinkurl{10.1007/s10458-023-09610-z}}
	
	
\end{thebibliography}

\end{document}